\documentclass{llncs}

\usepackage{amsmath}
\usepackage{amssymb}
\usepackage{microtype}

\DeclareMathOperator{\tw}{tw}

\begin{document}

\title{Tree-width of hypergraphs and surface duality}

\author{Fr\'ed\'eric Mazoit%
  \thanks{Research supported by the french ANR-project "Graph
    decompositions and algorithms (GRAAL)".}}

\institute{LaBRI Universit\'e Bordeaux,\\
  351 cours de la Lib\'eration F-33405 Talence cedex, France\\
  \email{Frederic.Mazoit@labri.fr}}

\maketitle

\begin{abstract}
  In Graph Minor III, Robertson and Seymour conjecture that We prove
  that given a hypergraph $H$ on a surface of Euler genus $k$, the
  tree-width of $H^*$ is at most the maximum of $\tw(H)+1+k$ and the
  maximum size of a hyperedge of $H^*$.
\end{abstract}

\section{Preliminaries}
A \emph{surface} is a connected compact 2-manyfold without boundaries.
A surface $\Sigma$ can be obtained, up to homeomorphism, by adding
$k(\Sigma)$ ``crosscaps'' to the sphere. $k(\Sigma)$ is the
\emph{Euler genus} or just \emph{genus} of the surface.

Let $\Sigma$ be a surface.  A graph $G=(V,E)$ on $\Sigma$ is a drawing
of a graph in $\Sigma$, i.e. each vertex $v$ is an element of
$\Sigma$, each edge $e$ is an open curve between two vertices, and
edges are pairwise disjoint. We only consider graphs up to
homomorphism. A face of $G$ is a connected component of
$\Sigma\setminus G$. We denote by $V(G)$, $E(G)$ and $F(G)$ the
vertex, edge and face sets of $G$.  We only consider \emph{2-cell}
graphs, i.e. graph whose faces are homeomorphic to open discs. The
Euler formula links the number of vertices, edges and faces of a graph
$G$ to the genus of the surface
\begin{equation*}
  |V(G)|-|E(G)|+|F(G)|=2-k(G).
\end{equation*}
The set $A(G)=V(G)\cup E(G)\cup F(G)$ of \emph{atoms} of $G$ is a
partition of $\Sigma$. Two Atom $x$ and $y$ of $G$ are \emph{incident}
if $x\cap\bar y$ or $y\cap \bar x$ is non empty, $\bar z$ being the
closure of $z$.  A \emph{cut-edge} in a graph $G$ on $\Sigma$ is an
edge $e$ separates $G$, i.e. $G$ intersects at least two connected
components of $\Sigma\setminus \bar e$. As an example, if a planar
graph $G$ has a cut-vertex $u$, any loop on $u$ that goes ``around'' a
connected component of $G\setminus\{u\}$ is a cut-edge.

Let $G=(V\cup V_E,L)$ be a bipartite graph on $\Sigma$. The graph $G$
can be seen as the incidence graph of a hypergraph. For each $v_e\in
V_E$, merge $v_e$ and its incident edges into a \emph{hyperedge} $e$,
and call $v_e$ its \emph{center}. Let $E$ be the set of all
\emph{hyperedges}. A \emph{hypergraph on $\Sigma$} is any such pair
$H=(V,E)$. For brevity, we also say \emph{edges} for hyperedges.  We
extend the notions of cut-edges, 2-cell graphs, atoms and incidence to
hypergraphs. Moreover, since they naturally correspond to abstract
graphs and hypergraphs, graph and hypergraph on surface inherit
terminology from them. For example, we denote $|e|$ the number of
vertices incident to a hyperedge $e$, and we denote $\alpha(H)$ the
maximum size of an edge of $H$. Note that a graph on $\Sigma$ is also
a hypergraph on $\Sigma$.

The dual of a hypergraph $H=(V,E)$ on $\Sigma$ is obtained by choosing
a vertex $v_f$ for every face $f$ of $H$. For every edge $e$ of center
$v_e$, we pick up an edge $e^*$ as follows: choose a local orientation
of the surface around $v_e$. This local orientation induces a cyclic
order $v_1$, $f_1$, $v_2$, $f_2$, \dots, $v_d$, $f_d$ of the ends of
$e$ and of the faces incident with $e$ (possibly with repetition). The
edge $e^*$ is the edge obtained by ``rotating'' $e$ and whose ends are
$v_{f_1}$, \dots, $v_{f_d}$.

A \emph{tree-decomposition} of a hypergraph $H$ on $\Sigma$ is a pair
$\mathcal{T}=(T,(X_v)_{v\in V(T)})$ with $T$ a tree and $(X_v)_{v\in
  V(T)}$ a family of \emph{bags} such that:
\begin{enumerate}
\item[i.] $\bigcup_{v\in V(T)}X_v=H$;
\item[ii.] $\forall x$, $y$, $z\in V(T)$ with $y$ on the path from $x$
  to $z$, $X_x\cap X_z\subseteq X_y$.
\end{enumerate}
The \emph{width} of $\mathcal{T}$ is
$\tw(\mathcal{T})=\max\bigl(|V(X_t)|-1\;;\; t\in V(T)\bigr)$ and the
\emph{tree-width} $\tw(H)$ of $H$ is the minimum width of one of its
tree-decompositions.

Tree-width was introduced by Robertson and Seymour in connection with
graph minors.  In~\cite{RoSe84a}, they conjectured that for a planar
graph $G$, $\tw(G)$ and $\tw(G^*)$ differ by at most one. In an
unpublished paper, Lapoire~\cite{La96b} proves a more general result:
for any hypergraph $H$ in an orientable surface $\Sigma$,
$\tw(H^*)\leq\max(\tw(H)+1+k(\Sigma),\alpha(H^*)-1)$.  Nevertheless,
his proof is rather long and technical. Later, Bouchitt\'e et
al.~\cite{BoMaTo03a} gave an easier proof for planar graphs. Here we
generalises Lapoire's result to arbitrary surfaces while being less
technical.

To avoid technicalities, we suppose that $H$ is connected, contains at
least two edges, has no pending vertices (i.e. vertices incident with
only one edge) and no cut-edge.

\section{P-trees and duality}

From now on, $H=(V,E)$ is a hypergraph on a surface $\Sigma$. The
\emph{border} of a partition $\mu$ of $E$ is the set of vertices
$\delta(\mu)$ that are incident with edges in at least two parts of
$\mu$, and the border of $X\subseteq E$ is the border of the partition
$\{X,E\setminus X\}$. A partition $\mu=\{X_1, \dots, X_p\}$ of $E$ is
\emph{connected} if there is a \emph{connecting partition} $\{V_1,
X_1, F_1, \dots, V_p, X_p, F_p\}$ of $A(H)\setminus\delta(\mu)$ so
that each $V_i\cup X_i\cup F_i$ is connected in $\Sigma$.

A \emph{p-tree} of $H$ is a tree $T$ whose internal nodes have degree
three and whose leaves are labelled with the edges of $H$ in a
bijective way. Removing an internal node $v$ of $T$ results in a
partition $\mu_v$ of $E$. Labelling each internal node $v$ of $T$ with
$\delta(\mu_v)$, turns $T$ into a tree-decomposition. The
\emph{tree-width} of a p-tree is its \emph{tree-width}, seen as a
tree-decomposition.  A p-tree is \emph{connected} if all its nodes
partitions are connected.

Let $\{A,B\}$ be a connected bipartition of $H$ and $\{V_A, A, F_A,
V_B, B, F_B\}$ a corresponding connecting partition. We define a
\emph{contracted} hypergraph $H/A$ as follows. Consider the incidence
graph $G_H(V\cup V_E,L)$ of $H$, and identify the edges in $A$ with
their centers. By adding edges trough faces in $F_A$, we can make
$G_H[A\cup V_A]$ connected. We then contract $A\cup V_A$ into a single
edge center $v_A$.  To make the resulting graph bipartite, we remove
all $v_A$-loops. When removing a loop $e$ incident to only one face
$F$, the new face $F\cup e$ is not a disc but a crosscap. Since the
border of $F\cup e$ is a loop, we can ``cut'' $\Sigma$ along this loop
and replace $F\cup e$ by an open disc while decreasing the genus of
the surface. The obtained graph is the bipartite graph of $H/A$. A
connected partition $\{A,B\}$ is \emph{non trivial} if neither $H/A$
nor $H/B$ are equal to $H$.

We need the following folklore lemma:
\begin{lemma}\label{lem:splitting}
  For any connected bipartition $\{A,B\}$ of $H$,
  $\tw(H)\leq\max\bigl(\tw(H/A),\break\tw(H/B)\bigr)$.  If
  $\delta(\{A,B\})$ belongs to a bag of an optimal tree-decomposition,
  then $\tw(H)=\max\bigl(\tw(H/A),\tw(H/B)\bigr)$.
\end{lemma}

Let $S$ be a set of vertices of $H$. An \emph{$S$-bridge} is a minimal
subset $X$ of $E$ with the property that $\delta(X)\subseteq S$. There
are two kind of $S$-bridges: singletons containing an edge whose ends
all belong to $S$ and sets $E_C$ containing all the edges incident to
at least one vertex in $C$, a connected component of $G\setminus S$.
The $S$-bridges partition $E$. We define the abstract graph $G_{/S}$
whose vertices are the $S$-bridges and in which $\{X,Y\}$ is an edge
if there is a face incident with both an edge in $X$ and an edge in
$Y$. A key fact is that any bipartition $\{A,B\}$ of $V(G_{/S})$ such
that $G_{/S}[A]$ and $G_{/S}[B]$ is connected corresponds to the
connected bipartition $\{\cup A, \cup B\}$.

\begin{proposition}\label{prop:1}
  There exists a connected p-tree $T$ of $H$ with $\tw(T)=\tw(H)$.
\end{proposition}
\begin{proof}
  By induction on $|E|$, if $|E|\leq 3$, since $H$ has no cut-edge,
  the only p-tree is connected and optimal. We can suppose that
  $|E|\geq 4$. We claim that there exists a connected non trivial
  bipartition $\{A,B\}$ of $E$ whose border is contained in a bag of
  an optimal tree-decomposition of $H$. Two cases arise:
  \begin{itemize}
  \item If the trivial one vertex tree-decomposition whose bag is $H$
    is optimal, we consider the graph $G_{/V}$. Since they are in
    bijection with the edges of $H$, and since $H$ has no cut edge,
    $G_{/V}$ has at least four vertices and no cut vertex. There thus
    exists a bipartition $\{A,B\}$ of $V(G_{/V})$ with $|A|,|B|\geq
    2$, $G_{/V}[A]$ and $G_{/V}[B]$ connected which gives a connected
    non trivial bipartition of $E$.
  \item Otherwise, there exists a separator $S$ contained in a bag of
    an optimal tree-decomposition of $H$. Let $C$ and $D$ be two
    connected component of $H\setminus S$, and $S_C$ and $S_C$ their
    corresponding $S$-bridges. Since $H$ contains no pending vertex,
    $|S_C|, |S_D|\geq 2$. Let $x$ and $y$ be the vertices of $G_{/S}$
    corresponding to $S_C$ and $S_D$. Take a spanning tree of
    $G_{/S}$. Removing an edge between $x$ and $y$ leads to a
    connected non-trivial bipartition of $E$, which finishes the proof
    of the claim.
  \end{itemize}

  Since $\{A, B\}$ is connected, $e_A$ and $e_B$ are respectively not
  cut-edges in $H/A$ and $H/B$. By induction, there exists connected
  p-trees $\mathcal{T}_A$ and $\mathcal{T}_B$ of optimal width of
  $H/A$ and $H/B$. By removing the leaves labelled $e_A$ and $e_B$ and
  adding an edge between their respective neighbour, we obtain from
  $\mathcal{T}_A\sqcup\mathcal{T}_B$ a p-tree of $H$ which is
  connected. Its width is $\max(\tw(\mathcal{T}/A),
  \tw(\mathcal{T}/B))$ which is equal, by Lemma~\ref{lem:splitting} to
  $\tw(H)$.\qed
\end{proof}

Because of the natural bijection between $E(H)$ and $E(H^*)$, a p-tree
$T$ of $H$ also corresponds to a p-tree $T^*$ of $H^*$.
\begin{proposition}\label{prop:2}
  For any connected p-tree $T$ of $H$,
  \begin{equation*}
    \tw(T^*)\leq\max(\tw(T)+1+k(\Sigma),\alpha(H^*)-1).
  \end{equation*}
\end{proposition}
\begin{proof}
  Let $v$ be a vertex of $T$ labelled $X_v$ in $T$ and $X^*_v$ in
  $T^*$.  If $v$ is a leaf, then $X^*_v=\{e^*\}$ and
  $|X^*_v|-1\leq\max(\tw(T)+1+k(\Sigma),\alpha(H^*)-1)$.  Otherwise,
  let $\{A, B, C\}$ be the $E$-partition associated to $v$. The label
  of $v$ in $T$ and $T^*$ is respectively $X_v=\delta(\{A,B,C\})$ and
  $X^*_v$, the set of faces incident with edges in at least two parts
  among $A$, $B$ and $C$.

  As for the proof of Proposition~\ref{prop:1}, since $\{A, B, C\}$ is
  connected, we may contract $A$ (and $B$ and $C$). But since we now
  care about the faces of $H$, we have to be more careful. We want an
  upper bound on $|X^*_v|$, we may thus add but not remove faces to
  $X^*_v$. So adding edges to make $G_H[A\cup V_A]$ connected is OK,
  but we cannot remove a loop $e$ on say $v_A$ incident with two faces
  in $X^*_v$. Instead, we cut $\Sigma$ along $e$ and fill the holes
  with open discs. While doing so, we removed $e$, we cut $v_A$ in two
  \emph{siblings}, and we decreased the genus of $\Sigma$.

  After contracting $A$, $B$ and $C$, we obtain a bipartite graph
  $G_v$ on $\Sigma'$ that has $|X_v|+3+s$ vertices with $s$ the number
  of siblings, at least $|X^*_v|$ faces and with $k(\Sigma')\leq
  k(\Sigma)-s$. Since $G_v$ is bipartite and faces in $X^*_v$ are
  incident with at least 4 edges, $2|E(G_v)|=4|F_4|+6|F_6|+\dots\geq
  4|F(G_v)|$ with $F_{2k}$ the set of $2k$-gones faces of $G_v$, and
  thus $|E(G_v)|\geq 2|F(G_v)|$. If we apply Euler's formula to $G_v$
  on $\Sigma'$, we obtain:
  $|X_v|+3+s-|E(G_v)|+|F(G_v)|=2-k(\Sigma')\geq 2-k(\Sigma)+s$. Adding
  this to $|E(G_v)|\geq 2|F(G_v)|$, we get $|X_v|+1+k(\Sigma)\geq
  |F(G_v)|\geq|X^*_v|$ which proves that
  $|X^*_v|-1\leq\max(\tw(T)+1+k(\Sigma),\alpha(H^*)-1)$, and thus
  $\tw(T^*)\leq\max(\tw(T)+1+k(\Sigma),\alpha(H^*)-1)$.\qed
\end{proof}

Let us now prove the main theorem.
\begin{theorem}
  For any hypergraph $H$ on a surface $\Sigma$,
  \begin{equation*}
    \tw(H^*)\leq\max\bigl(\tw(H)+1+k(\Sigma),\alpha(H^*)-1\bigr).
  \end{equation*}
\end{theorem}
\begin{proof}
  By Proposition~\ref{prop:1}, let $T$ be a connected p-tree of $H$
  such that $\tw(T)=\tw(H)$. By Proposition~\ref{prop:2},
  $\tw(T^*)\leq\max(\tw(T)+1+k(\Sigma),\alpha(H^*)-1)$. Since
  $\tw(H^*)\leq\tw(T^*)$, we deduce,
  $\tw(H^*)\leq\max(\tw(H)+1+k(\Sigma),\alpha(H^*)-1)$.\qed
\end{proof}

\end{document}